\documentclass[a4paper]{article}
\usepackage{amsfonts}
\usepackage{amsthm}
\usepackage{amsmath}
\usepackage{amssymb}

\usepackage{tikz}
\usepackage{pgfpages}

\usepackage{url}
\usepackage{color}

\usepackage{skak}

\usetikzlibrary{patterns}

\newtheorem{theorem}{Theorem}
\newtheorem{lemma}[theorem]{Lemma}

\newtheorem{proposition}[theorem]{Proposition}

\pgfdeclarepatternformonly{north west lines2}{%
\pgfqpoint{-1pt}{-1pt}}{\pgfqpoint{4pt}{4pt}}{\pgfqpoint{3pt}{3pt}}%
{
  \pgfsetlinewidth{0.4pt}
  \pgfpathmoveto{\pgfqpoint{0pt}{3pt}}
  \pgfpathlineto{\pgfqpoint{3.1pt}{-0.1pt}}
  \pgfusepath{stroke}
}

\newcommand\Mrook[2]{
 \node at (#1+0.5,#2+0.45) {\BlackRookOnWhite};
}
\newcommand\Mknight[2]{
 \node at (#1+0.5,#2+0.45) {\BlackKnightOnWhite};
}
\newcommand\Mpawn[2]{
 \node at (#1+0.5,#2+0.45) {\BlackPawnOnWhite};
}
\newcommand\Mbishop[2]{
 \node at (#1+0.5,#2+0.45) {\BlackBishopOnWhite};
}
\newcommand\Wbishop[2]{
 \node at (#1+0.5,#2+0.45) {\WhiteBishopOnWhite};
}

\newcommand{\true}{\mathtt{true}}
\newcommand{\false}{\mathtt{false}}

\begin{document}

\title{Solitaire Chess is NP-complete}

\author{Jens Ma{\ss}berg\\ \small
 Institut f\"ur Optimierung und Operations
Research,
Universit\"at Ulm,\\
\small  {jens.massberg@uni-ulm.de}}

\definecolor{lightgrayX}{rgb}{0.9, 0.9, 0.9}
\definecolor{lightgray}{rgb}{0.8, 0.8, 0.8}

\definecolor{colorA}{HTML}{ED7F22}
\definecolor{colorB}{HTML}{2D3E50}
\definecolor{colorC}{HTML}{297FB8}
\definecolor{colorD}{HTML}{C1392B}

\maketitle


\begin{abstract}
 ``Solitaire Chess'' is a logic puzzle published by Thinkfun\texttrademark\, 
that 
can be seen as a single person version of traditional chess.
Given a chess board with some chess pieces of the same color placed on it, the 
task is to 
capture all pieces but one using only moves that are allowed in chess.
Moreover, in each move one piece has to be captured.
We prove that deciding if a given instance of Solitaire Chess is solvable is
NP-complete.
\end{abstract}

\section{Introduction}

Solitaire Chess is a one-player puzzle game published by 
Thinkfun\texttrademark\,
\cite{thinkfun}. The puzzle is based on chess and can be seen as a single 
person version of it.
As in the original chess game, there are six different types of chess pieces:
king, queen, rook, bishop, knight and pawn.

In each instance some of the pieces (where several pieces of the same type are 
allowed) are placed on a $4\times 4$ chess 
board. Now the task is to capture all but one piece using only moves that are 
allowed in the original chess game (see e.g. \cite{chesswiki}). An additional 
rule is, that in each move 
one piece has to be captured.
Figure \ref{fig:example} shows an example instance and a sequence of moves 
solving it.

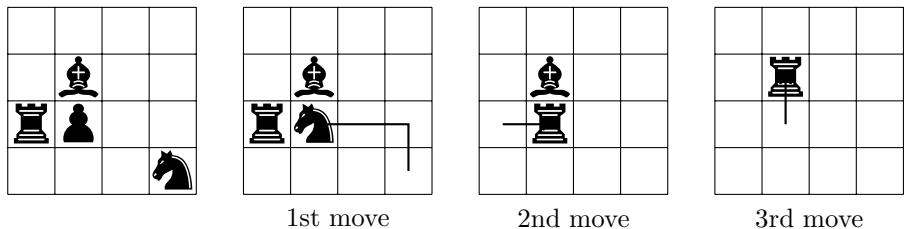
\begin{figure}[ht]
\begin{center}
\begin{tikzpicture}[scale=0.62]
 \draw[step=1.0,black] (0,0) grid (4,4);
 \Mrook{0}{1};
 \Mpawn{1}{1};
 \Mbishop{1}{2};
 \Mknight{3}{0};

 \begin{scope}[shift={(5,0)}]
  \draw[step=1.0,black] (0,0) grid (4,4);
  \Mrook{0}{1};
  \Mbishop{1}{2};
  \Mknight{1}{1};
  \draw[->,thick] (3.5,0.5) -- ++ (0,1) -- ++ (-2,0);
  
  \node at (2,-0.5) {1st move};
 \end{scope}

 \begin{scope}[shift={(10,0)}]
  \draw[step=1.0,black] (0,0) grid (4,4);
  \Mrook{1}{1};
  \Mbishop{1}{2};
  \draw[->,thick] (0.5,1.5) -- ++ (1,0);

  \node at (2,-0.5) {2nd move};
 \end{scope}

 \begin{scope}[shift={(15,0)}]
  \draw[step=1.0,black] (0,0) grid (4,4);
  \Mrook{1}{2};
  \draw[->,thick] (1.5,1.5) -- ++ (0,1);

  \node at (2,-0.5) {3rd move};
\end{scope}

\end{tikzpicture}
\caption{Example instance and a sequence of moves solving it.}
\label{fig:example}
\end{center}
\end{figure}

To our knowledge the puzzle has been first published in 2010 and
is available as hardware game, iPhone and Android apps and as online puzzle.

In this paper we consider a generalized version of the puzzle where an instance 
consists of $n$  pieces on a chess board of size $N\times N$, $N\in\mathbb{N}$. 
Moreover, every type of piece might appear several times in an instance.
 We prove that this generalized version of Solitaire Chess is NP-hard by 
giving a polynomial reduction from 3-SAT.
Note that the corresponding generalization of the traditional chess is 
EXPTIME-complete 
\cite{chess}.

\section{First Observations}

We identify each square of the board by a position $(i,j)\in \mathbb{Z}^2$.
If $i+j \equiv 0 \text{ mod } 2$, we call the square $(i,j)$ \emph{white}. 
Otherwise, it is called \emph{black}.
 The orientation of the chess board is only important for pawns, as they 
are only allowed to move ``forward''. As a moving piece must 
capture 
another piece, a pawn on position $(i,j)$ can only capture a piece at position 
$(i-1,j+1)$ or $(i+1,j+1)$.

A move is uniquely defined by a pair of pieces $(p_1,p_2)$ where $p_1$ captures 
$p_2$.  A move $(p_1,p_2)$ is feasible if $p_1$ can capture $p_2$ following the
rules of chess (ignoring the color of the pieces).

\begin{lemma}\label{lemma:inNP}
 Solitaire Chess is in NP. 
\end{lemma}
\begin{proof}
 The feasibility of a move can be checked in linear time.
 An instance with $n$ chess pieces can be solved, if there exists a sequence of 
 moves such that only one piece survives. 
 As the number of pieces is reduced by one in each move such a sequence consists
 of exactly  $n-1$ moves.
 Thus the coding length of such a sequence is linear in
 the number of pieces and the feasibility of all moves can be verified in
polynomial time,
 implying that the problem is in NP.
\end{proof}

\section{Reduction from 3-SAT}

3-Satisfiability (of in short 3-SAT) is a well known problem in complexity 
theory.
A 3-SAT instance consist of $n$ variables $x_1,\ldots, x_n$ and $m$ 
clauses $C_1,\ldots, C_m$ where each clause contains exactly three literals of
$\{x_1,\ldots, x_n, \overline{x_1},\ldots,\overline{x_n}\}$.
For a given instance $I$ the problem is to determine if there exists a truth 
assignment $\tau:\{x_1,\ldots, x_n\}\rightarrow \{\true,\false\}$ such that at 
least one literal of each clause is satisfied.
It is well known that 3-SAT is NP-complete \cite{cook}.

Before defining the reduction from 3-SAT to Solitaire Chess in detail we give a 
short overview of our construction.
For every variable we add two variable columns
corresponding to the two possible truth assignments $\true$ and $\false$ for 
this variable.
 A \emph{variable rook} will move in one of these rows.
For every clause we add three clause rows, one for each literal of 
the clause. 
Moreover, we have two rooks for every clause that can move in two of 
the three rows. The third row corresponds to the literal of the clause which 
has to be satisfied by the truth assignment we are looking for.
In order to interlink the columns and rows we add for every 
literal of every clause two literal bishops. They are placed on the 
corresponding clause row and variable column, respectively, such that they can 
capture each other. If and only if all literal bishops can be captured, the
initial 3-SAT instance is satisfiable.
Finally, we require an additional bishop and several pawns in order to 
guarantee that all rooks can be captured.

In our construction we have to take care for the rooks not to leave their 
clause rows or variable columns. Moreover, the literal bishops should be only 
able to capture their partner literal bishop and no other piece.

For our reduction we set $M=8m^2$.

\subsection{Variables}

Let $x_i$, $i\in\{1,\ldots, n\}$, be a variable.
The columns $iM$ and $iM+2$ are the \emph{variable columns} of $x_i$ that 
correspond to the truth assignment of $x_i$, that is,
they correspond to $x_i=\true$ and $x_i=\false$, respectively.
Moreover, we have a \emph{variable rook} $r$ for $x_i$ which is originally 
placed in one of the two columns.

First the rook must pass a  \emph{column changing gate} consisting of three 
pawns $p_1,p_2,p_3$ as shown in Figure \ref{fig:variable}.
Depending in which direction the rook passes the gate, it will end in the left 
or the right column.
If the rook executes the moves $(r,p_1), (r,p_2), (r,p_3)$, it ends in the 
right column.
Otherwise the rook ends in the left column after executing the moves  $(r,p_3), 
(r,p_2)$ and $(r,p_1)$.
Note, that the three pawns cannot capture other pieces and have to be captured 
by $r$.
The remaining pieces are placed in such a way that a variable rook can leave
its column neither after passing the column changing gate nor after capturing a 
piece that does not belong to the gate.

Initially, the variable rook for $x_i$ is placed at position $(iM+2, 5m^2+ 8i)$ 
and the pawns of the columns changing gate at $(iM,5m^2+8i), (iM,5m^2+8i-2)$ and
$(iM+2,5m^2+8i-2)$.

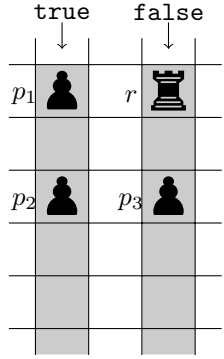
\begin{figure}[ht]
 \begin{center}
  \begin{tikzpicture}[scale=0.7]
    \draw[->] (0.5,3.8) -- ++(0,-0.5);
    \node at (0.5,4) {$\true$};
    \draw[->] (2.5,3.8) -- ++(0,-0.5);
    \node at (2.5,4) {$\false$};

    \draw[fill=lightgray, draw=white] (0,-2.5) rectangle (1,3);
    \draw[fill=lightgray, draw=white] (2,-2.5) rectangle (3,3);
    \draw[step=1.0,black] (-0.5,-2.5) grid (3.5,3.5);
    \Mrook{2}{2};
    \Mpawn{0}{2};
    \Mpawn{2}{0};
    \Mpawn{0}{0};
\node at (2-0.2,2.4) {$r$};
\node at (0-0.2,2.4) {$p_1$};
\node at (0-0.2,0.4) {$p_2$};
\node at (2-0.2,0.4) {$p_3$};
  \end{tikzpicture}
\caption{A variable rook, a column changing gate and two variable columns.}
\label{fig:variable}
\end{center}
\end{figure}

\subsection{Clause}

Let $C_j=(l_j^1 \vee l_j^2 \vee l_j^3)$, $j\in\{1,\ldots, m\}$, be a clause.
The three rows $6mj+2, 6mj+4$ and $6mj+6$ correspond to the literals $l_j^1$, 
$l_j^2$ and $l_j^3$, respectively.
Moreover we have two \emph{clause rooks} $r_1$ and $r_2$ for clause 
$C_j$ originally placed in the rows $6mj+6$ and $6mj+4$, respectively.
They will represent the two literals that might be unsatisfied by a 
feasible truth assignment.
To this end, they must be able to change the row where they are initially
placed.
We introduce a \emph{row changing gate} consisting of three pawns $p_1,p_2$ and 
$p_3$ that works as the column changing gate for variable rooks:
If the rook $r_1$ executes the moves 
$(r_1,p_1),(r_1,p_2)$ and $(r_1,p_3)$ it remains in its original row.
Otherwise it can change the row by executing the moves 
$(r_1,p_3),(r_1,p_2)$ and 
$(r_1,p_1)$. Note that pawns of a row changing gate have to 
be captured by clause rooks.
A second row changing gate is required for the second rook.
Figure \ref{fig:clause} shows the initial placement of the clause rooks and 
pawns of the row changing gate.

By using the row changing gates we can assure that the two rooks are in any two of
the three clause rows.
Note that it is also possible, that one of the rooks captures the other one.
In this case the instance can be solved only if there exists a truth assignment 
such that at least two of the literals of the clause are satisfied. 

Initially the positions of the clause rooks for clause $C_j$ are $(Mn+10j , 
6mj+4)$ and $(Mn+10j-4, 6mj+6)$.

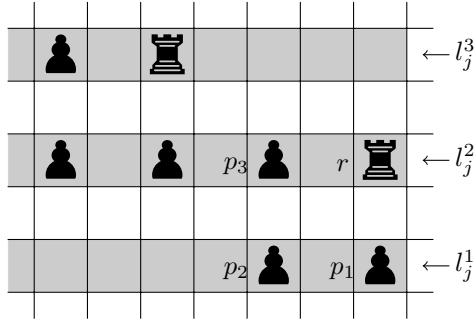
\begin{figure}[ht]
\begin{center}
  \begin{tikzpicture}[scale=0.7]
    \draw[->] (7.8,-1.5) -- ++(-0.5,0);
    \node at (8.1,-1.5) {$l^1_j$};
    \draw[->] (7.8,0.5) -- ++(-0.5,0);
    \node at (8.1,0.5) {$l^2_j$};
    \draw[->] (7.8,2.5) -- ++(-0.5,0);
    \node at (8.1,2.5) {$l^3_j$};
  
    \draw[fill=lightgray, draw=white] (-0.5,-2) rectangle (7,-1);
    \draw[fill=lightgray, draw=white] (-0.5,0) rectangle (7,1);
    \draw[fill=lightgray, draw=white] (-0.5,2) rectangle (7,3);
    \draw[step=1.0,black] (-0.5,-2.5) grid (7.5,3.5);
    \Mrook{6}{0};
    \Mrook{2}{2};
    \Mpawn{2}{0};
    \Mpawn{0}{0};
    \Mpawn{0}{2};

    \Mpawn{6}{-2};
    \Mpawn{4}{-2};
    \Mpawn{4}{0};

   \node at (6-0.2,0.4) {$r$};
   \node at (6-0.2,0.4-2) {$p_1$};
   \node at (4-0.2,0.4-2) {$p_2$};
   \node at (4-0.2,0.4) {$p_3$};
  \end{tikzpicture}
\end{center}
\caption{Clause rows of a clause $C_j$ and the initial placement of the two 
clause rooks and two row changing gates.}
\label{fig:clause}
\end{figure}

\subsection{Literal Bishops}

Up to now we have defined the pieces, rows and columns required to represent 
clauses and variables. Now we interlink the clauses and the corresponding 
literals.
Let $C_j$, $j\in\{1,\ldots, m\}$, be a clause and $l_j^k$, $k\in\{1,2,3\}$, be 
one of its literals. Assume $l_j^k$ is a literal of the variable $x_i$,
$i\in\{1,\ldots,n \}$.
We add a pair of \emph{literal bishops} $b_1$ and $b_2$ which interlink the 
$k$'th clause row of clause $C_j$ and the left or right variable column of 
$x_i$ (depending on $l_j^k$ being positive or negative).
We place $b_1$ and $b_2$ such that $b_1$ can only be captured by $b_2$ or one 
of the clause rooks of $C_j$ and $b_2$ can be captured by $b_1$ or the literal
rook of $x_i$.

We place $b_1$ at position
\begin{equation}
(iM-7-2( i+j \text{ mod } m)+2k, 6mj+2k)  
\end{equation}
and $b_2$ at position
\begin{eqnarray}
  (iM,6mj +7 +2( i+j \text{ mod } m)) &&\text{if literal }l_j\text{ is 
positive and} \\
  (iM+2,6mj +9 +2( i+j \text{ mod } m))) &&\text{if literal }l_j\text{ is 
negative}.
\end{eqnarray}

\begin{lemma}\label{lemma:lit}
 The only feasible move for a literal bishop is to capture a piece at the 
position of its opponent.
\end{lemma}
\begin{proof}
 First note that the literal bishops are on black squares and all other pieces
are on white squares. 
Thus a literal bishop can only capture pieces that are currently
placed on the original positions of other literal bishops.

Let $b$ be a literal bishop for a literal of clause $C_j$ and of variable 
$x_i$ and $b'$ be a literal bishop for a literal of clause $C_{j'}$ and 
variable $x_{i'}$ with $(i,j)\neq (i',j')$.
Let $(x,y)$ and $(x',y')$ be the positions of $b$ and $b'$, respectively.
We want to prove, that $b$ cannot capture $b'$. It suffices to show that 
$x-y\neq x'-y'$ and $x+y\neq x'+y'$.

To prove the first inequality observe that
$$ (i-1)M < x-y < iM.$$
Thus if $i\neq i'$ then $x-y\neq x'-y'$.
If $i=i'$ then 
$$x-y = iM -6mj-7-2 (i+j\text { mod }m))\equiv -7-2(i+j\text { mod }m)\text{
mod } 6m$$
and therefore $x-y\neq x'-y'$ if $j\neq j'$.

To prove the second inequality observe that
$$iM+6mj-3m < x+y < iM+6mj+3m.$$
Thus $\lfloor (x+y+3m)/ M\rfloor = i$ and $ \lfloor((x+y \text{ mod 
}M)+3m)/6m\rfloor = j$ implying $x+y\neq x'+y'$ if $(i,j)\neq(i',j')$.

Now let $b_1$ and $b_2$ be the two literal bishops of the same literal
at position $(x_1,y_1)$ and $(x_2,y_2)$, respectively. By definition
$x_1-y_1=x_2-y_2$, thus $b_1$ can capture $b_2$ and vice versa.
\end{proof}

\begin{figure}[ht]
\begin{center}
  \begin{tikzpicture}[scale=0.7]
    \draw[fill=lightgray, draw=white] (4,-2.5) rectangle (5,6.5);
    \draw[fill=lightgray, draw=white] (6,-2.5) rectangle (7,6.5);
    \draw[fill=lightgray, draw=white] (-1.5,-2) rectangle (7.5,-1);
    \draw[fill=lightgray, draw=white] (-1.5,0) rectangle (7.5,1);
    \draw[fill=lightgray, draw=white] (-1.5,2) rectangle (7.5,3);
    \draw[step=1.0,black] (-1.5,-2.5) grid (7.5,6.5);

    \Wbishop{6}{5};
    \Wbishop{4}{3};

    \Mbishop{3}{2};
    \Mbishop{1}{0};
    \Mbishop{-1}{-2};

  \end{tikzpicture}
 \end{center}
\caption{Possible positions for the two literal bishops of a literal. The black
bishops indicate feasible positions for $b_1$ and the white bishops positions 
for
$b_2$.}
\end{figure}
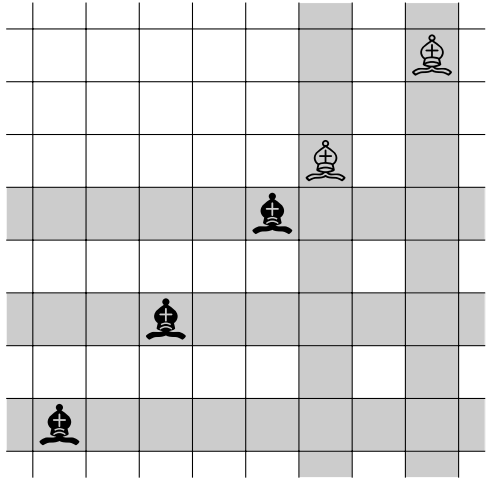


\subsection{Cleaning Pieces}
Finally, we require some more pieces whose purpose is to capture all 
remaining pieces in a feasible instance.

To this end we add a \emph{cleaning bishop} at position $(0,0)$.
Moreover, we add for every clause row $y$ a pawn at position $(-y,y)$ 
and for every variable row $x$ a pawn at position $(x,-x)$. Thus we add in 
total $3m+2n$ \emph{cleaning pawns}.
We have already seen, that the clause and variable rooks can capture all pawns 
that belong to row and column changing gates. Moreover, each rook can capture 
one of the cleaning pawns. After this, the cleaning bishop can capture all 
pawns 
and rooks at positions $(i,-i)$, $i\in\mathbb{Z}$. Thus if there are no literal 
bishops, we can capture all but one piece and the instance is solvable.

So the challenge of such an instance is to capture all literal bishops.

\subsection{NP-completeness}

Before finally proving that the transformed Solitaire Chess instance is 
solvable if and only if the original 3-SAT instance is satisfiable, we need 
some more observations, which pieces can be captured by other ones.
First note, that literal bishops cannot capture cleaning pawns or the cleaning
bishop:
Initially all literal bishops are on black squares, while all pawns and 
the cleaning bishop are on white squares.

\begin{proposition}
 A clause rook can change its row only in row changing gates
and variable rooks can change their column only in their column changing gate.
\end{proposition}
\begin{proof}
 Let $r$ be a clause rook positioned in row $y$. In order to change its row
there must be two
 pieces at positions $(x,y)$ and $(x,y')$ for $x,y\in\mathbb{Z}$, $y\neq y'$.
 But by construction for any piece in row $y$ at position $(x,y)$ that does not 
 belong to a row changing gate there is no other piece with the same x-coordinate.
 The same arguments apply to variable rooks by exchanging rows and columns.
\end{proof}

\begin{proposition}
 If the Solitaire Chess instance is solvable, then in an optimal sequence of
moves at least one of the two literal bishops is captured by a variable or a
clause rook.
\end{proposition}
\begin{proof}
By Lemma \ref{lemma:lit} a literal bishop can capture only a piece at the 
position of the opponent bishop. But then this bishop has to be captured by a 
rook in the corresponding row or column.
\end{proof}

\begin{theorem}
 Solitaire Chess is NP-complete.
\end{theorem}
\begin{proof}
 The transformation is polynomial:
 For a 3-SAT instance with $n$ variables and $m$ clauses we require
 $17m+1+4n$ chess pieces and all pieces are placed on a board of polynomial
 size.

 Assume $I$ is a 3-SAT instance that is satisfied by a truth assignment $\tau$. 
We give a sequence of moves that solves the corresponding Solitaire 
Chess instance.
For every clause $C_j$, $j\in\{1,\ldots, m\}$, there exists at least one 
literal that is satisfied by $\tau$. We set $z_j\in \{1,2,3\}$ such that 
$l_j^{z_j}$ is satisfied by $\tau$.

For every variable $x_i$, $i\in\{1,\ldots, n\}$, the variable rook uses the 
column selecting gate in order to get into the row corresponding to $\tau(x_i)$.
The two rooks of every clause $C_j$, $j\in\{1,\ldots, m\}$, use the row 
selecting gates to get into the two rows corresponding to the two literals 
in $\{l_j^1,l_j^2,l_j^3\}\setminus\{l_j^{z_j}\}$.
Consider a literal $l_j^z$, $z\in\{1,2,3\}$.
If $z=z_j$, then the literal bishop $b_1$ captures 
$b_2$ and is captured by the variable rook. Otherwise, $b_2$ 
captures $b_1$ and is itself captured by the clause rook in the corresponding 
row.
By this all literal rooks are captures. Finally, the rooks capture their 
cleaning pawns and the cleaning bishops capture all remaining pieces. We
conclude that the Solitaire Chess instance is solvable.

Now assume that the Solitaire Chess instance is solvable.
We define a truth assignment $\tau$ by setting $\tau(x_i)=\true$ if the variable
rook corresponding to $x_i$ uses the right column and $\tau(x_i)=\false$
otherwise. 
Now consider a clause $C_j$, $j\in\{1,\ldots, m\}$. We have to show that at 
least one of the literals $l_j^1,l_j^2$ or $l_j^3$ is satisfied by $\tau$.
At most two of the corresponding literal bishop pairs have been captured by 
their 
clause rooks. Thus the third bishop pair must be captured by the corresponding 
variable rook. But then this literal is satisfied by $\tau$.

 We finish the proof by observing that by Lemma \ref{lemma:inNP} the problem is 
in
NP.
\end{proof}

\section*{Acknowledgment}

The author wants to thank the Brendel family for pointing him to Solitaire Chess
and Thinkfun\texttrademark\, for publishing the game.

\nocite{*}
\bibliography{solitairechess}{}
\bibliographystyle{plain}

\end{document}